\documentclass[10pt, twocolumn, aps, nofootinbib, longbibliography]{revtex4-1}
\pdfoutput=1	

\usepackage{amsfonts, amssymb, amsmath, amsthm}
\usepackage{latexsym}
\usepackage[tracking=smallcaps]{microtype}	
\usepackage{url}
\usepackage{color}
\definecolor{DarkGray}{rgb}{0.1,0.1,0.5}
\usepackage[colorlinks=true,breaklinks, linkcolor=DarkGray,citecolor=DarkGray,urlcolor=DarkGray]{hyperref}	

\usepackage{graphicx}
\usepackage[tight, TABBOTCAP]{subfigure}

\newcommand{\ket}[1]{{|#1\rangle}}

\newcommand{\identity}{\ensuremath{\boldsymbol{1}}} 

\def\CZ {C\!Z}	

\newcounter{sprows}

\newlength{\spheight}
\newlength{\spraise}

\newcommand{\comment}[1]{\emph{\color{blue}Comment:\color{black} #1}} 
\newlength{\commentslength}
\newcommand{\comments}[1]{
\hspace{-2\parindent}
\addtolength{\commentslength}{-\commentslength}
\addtolength{\commentslength}{\linewidth}
\addtolength{\commentslength}{-\parindent}
\fcolorbox{blue}{white}{\smallskip\begin{minipage}[c]{\commentslength}
\emph{Comments:}\begin{itemize}#1\end{itemize}\end{minipage}}\bigskip
}
\newcommand{\rem}[1]{}

\newtheorem{theorem}{Theorem}

\newtheorem{claim}[theorem]{Claim}

\newfont{\subsubsecfnt}{ptmri8t at 11pt}
\renewcommand{\subparagraph}[1]{\smallskip{\subsubsecfnt #1.}}

\newcommand{\eqnref}[1]{\hyperref[#1]{{(\ref*{#1})}}}
\newcommand{\thmref}[1]{\hyperref[#1]{{Theorem~\ref*{#1}}}}
\newcommand{\lemref}[1]{\hyperref[#1]{{Lemma~\ref*{#1}}}}
\newcommand{\corref}[1]{\hyperref[#1]{{Corollary~\ref*{#1}}}}
\newcommand{\defref}[1]{\hyperref[#1]{{Definition~\ref*{#1}}}}
\newcommand{\secref}[1]{\hyperref[#1]{{Sec.~\ref*{#1}}}}
\newcommand{\figref}[1]{\hyperref[#1]{{Fig.~\ref*{#1}}}}  
\newcommand{\tabref}[1]{\hyperref[#1]{{Table~\ref*{#1}}}}
\newcommand{\remref}[1]{\hyperref[#1]{{Remark~\ref*{#1}}}}
\newcommand{\appref}[1]{\hyperref[#1]{{Appendix~\ref*{#1}}}}
\newcommand{\claimref}[1]{\hyperref[#1]{{Claim~\ref*{#1}}}}
\newcommand{\factref}[1]{\hyperref[#1]{{Fact~\ref*{#1}}}}
\newcommand{\propref}[1]{\hyperref[#1]{{Proposition~\ref*{#1}}}}
\newcommand{\exampleref}[1]{\hyperref[#1]{{Example~\ref*{#1}}}}
\newcommand{\conjref}[1]{\hyperref[#1]{{Conjecture~\ref*{#1}}}}

\allowdisplaybreaks[1]

\def\COLOR{}
\ifdefined\COLOR

\else

\fi

\definecolor{Cayenne}{rgb}{0.5,0,0}
\definecolor{Midnight}{rgb}{0,0,0.5}
\definecolor{Plum}{rgb}{0.5,0,0.5}
\definecolor{Teal}{rgb}{0,0.5,0.5}
\definecolor{Clover}{rgb}{0,0.5,0}
\definecolor{Maroon}{rgb}{0.5,0,0.25}
\definecolor{Ocean}{rgb}{0,0.25,0.5}
\definecolor{Tangerine}{rgb}{1,0.5,0}
\definecolor{Strawberry}{rgb}{1,0,0.5}
\definecolor{Fern}{rgb}{0.25,0.5,0}
\definecolor{Aqua}{rgb}{0,0.5,1}
\definecolor{Moss}{rgb}{0,0.5,0.25}
\definecolor{Mocha}{rgb}{0.5,0.25,0}
\definecolor{Lemon}{rgb}{1,1,0}
\definecolor{Asparagus}{rgb}{0.5,0.5,0}
\definecolor{Grape}{rgb}{0.5,0,1}
\definecolor{Iron}{rgb}{.3,.3,.3}
\definecolor{Steel}{rgb}{.4,.4,.4}

\usepackage{enumitem} 

\makeatletter
\let\save@mathaccent\mathaccent
\newcommand*\if@single[3]{%
  \setbox0\hbox{${\mathaccent"0362{#1}}^H$}%
  \setbox2\hbox{${\mathaccent"0362{\kern0pt#1}}^H$}%
  \ifdim\ht0=\ht2 #3\else #2\fi
  }
\newcommand*\rel@kern[1]{\kern#1\dimexpr\macc@kerna}
\newcommand*\widebar[1]{\@ifnextchar^{{\wide@bar{#1}{0}}}{\wide@bar{#1}{1}}}
\newcommand*\wide@bar[2]{\if@single{#1}{\wide@bar@{#1}{#2}{1}}{\wide@bar@{#1}{#2}{2}}}
\newcommand*\wide@bar@[3]{%
  \begingroup
  \def\mathaccent##1##2{%
    \let\mathaccent\save@mathaccent
    \if#32 \let\macc@nucleus\first@char \fi
    \setbox\z@\hbox{$\macc@style{\macc@nucleus}_{}$}%
    \setbox\tw@\hbox{$\macc@style{\macc@nucleus}{}_{}$}%
    \dimen@\wd\tw@
    \advance\dimen@-\wd\z@
    \divide\dimen@ 3
    \@tempdima\wd\tw@
    \advance\@tempdima-\scriptspace
    \divide\@tempdima 10
    \advance\dimen@-\@tempdima
    \ifdim\dimen@>\z@ \dimen@0pt\fi
    \rel@kern{0.6}\kern-\dimen@
    \if#31
      \overline{\rel@kern{-0.6}\kern\dimen@\macc@nucleus\rel@kern{0.4}\kern\dimen@}%
      \advance\dimen@0.4\dimexpr\macc@kerna
      \let\final@kern#2%
      \ifdim\dimen@<\z@ \let\final@kern1\fi
      \if\final@kern1 \kern-\dimen@\fi
    \else
      \overline{\rel@kern{-0.6}\kern\dimen@#1}%
    \fi
  }%
  \macc@depth\@ne
  \let\math@bgroup\@empty \let\math@egroup\macc@set@skewchar
  \mathsurround\z@ \frozen@everymath{\mathgroup\macc@group\relax}%
  \macc@set@skewchar\relax
  \let\mathaccentV\macc@nested@a
  \if#31
    \macc@nested@a\relax111{#1}%
  \else
    \def\gobble@till@marker##1\endmarker{}%
    \futurelet\first@char\gobble@till@marker#1\endmarker
    \ifcat\noexpand\first@char A\else
      \def\first@char{}%
    \fi
    \macc@nested@a\relax111{\first@char}%
  \fi
  \endgroup
}
\makeatother

\def\beq{\begin{equation}}
\def\eeq{\end{equation}}

\def\llbracket{{[\![}}
\def\rrbracket{{]\!]}}

\renewcommand{\comment}[1]{}\renewcommand{\comments}[1]{}

\begin{document}
\def\compilefullpaper{}

\title{Quantum error correction with only two extra qubits}
\author{Rui Chao}
\author{Ben W. Reichardt}
\affiliation{University of Southern California}

\begin{abstract}
Noise rates in quantum computing experiments have dropped dramatically, but reliable qubits remain precious.  Fault-tolerance schemes with minimal qubit overhead are therefore essential.  We introduce fault-tolerant error-correction procedures that use only two ancilla qubits.  The procedures are based on adding ``flags" to catch the faults that can lead to correlated errors on the data.  They work for various distance-three codes.  

In particular, our scheme allows one to test the $\llbracket 5,1,3 \rrbracket$ code, the smallest error-correcting code, using only seven qubits total.  Our techniques also apply to the $\llbracket 7,1,3 \rrbracket$ and $\llbracket 15,7,3 \rrbracket$ Hamming codes, thus allowing to protect seven encoded qubits on a device with only $17$ physical qubits.  
\end{abstract}

\maketitle

\section{Introduction}

Fault-tolerant quantum computation is possible: quantum computers can tolerate noise and imperfections.  Fault tolerance will be necessary for running quantum algorithms on large-scale quantum computers.  However, fault-tolerance schemes have substantial overhead; many physical qubits are used for each encoded, logical qubit.  This means that on small- and medium-scale devices in the near future, it will be difficult to test fault tolerance, and to explore the efficacy of different fault-tolerance schemes.  We aim to reduce the qubit overhead, especially for small devices.  

Five-qubit systems are enough to test very limited fault-tolerance schemes~\cite{LinkeMonroe16errordetectionexperiment, Gottesman16smallexperiments}.  The $\llbracket 4, 1, 2 \rrbracket$ Bacon-Shor subsystem code encodes one logical qubit into four physical qubits.  A fifth qubit is needed for fault-tolerant error detection.  Since the code has distance two, it can only detect an error, not correct it.  

Until recently, the smallest known scheme that can correct an error used the $\llbracket 9,1,3 \rrbracket$ Bacon-Shor code, plus a tenth ancilla qubit.  Although smaller, more efficient error-correcting codes are known, such as the $\llbracket 5,1,3 \rrbracket$ perfect code (the smallest distance-three code), the $\llbracket 7,1,3 \rrbracket$ Steane code, or the $\llbracket 8,3,3 \rrbracket$ code, fault-tolerance schemes using these codes have required at least as many qubits total.  For example, Shor-style syndrome extraction~\cite{Shor96, DiVincenzoAliferis06slow} requires $w + 1$ or $w$ ancilla qubits, where $w$ is the largest weight of a stabilizer generator.  Steane-style error correction~\cite{Steane97, Steane02, DiVincenzoAliferis06slow} uses at least a full code block of extra qubits, and Knill's scheme~\cite{Knill03erasure} uses an encoded EPR state and thus at least two ancilla code blocks.  

Yoder and Kim~\cite{YoderKim16trianglecodes} have given fault-tolerance schemes using the $\llbracket 5,1,3 \rrbracket$ code with eight qubits total, and the $\llbracket 7,1,3 \rrbracket$ code with nine qubits total.  Extending the latter construction, we introduce fault-tolerance procedures that use only two extra qubits.  Figure~\ref{f:ecqubits} highlights some examples, and \tabref{f:ecancillassummary} compares the qubit overhead of our scheme to the Shor and Stephens-Yoder-Kim methods (described in \appref{s:decodedhalfcatsyndromeextraction}).  In particular, with the $\llbracket 5,1,3 \rrbracket$ code our scheme uses only seven qubits total, or ten qubits with the $\llbracket 8,3,3 \rrbracket$ code.  A particularly promising application is to the $\llbracket 15,7,3 \rrbracket$ Hamming code: $17$ physical qubits suffice to protect seven encoded qubits.  In~\cite{ChaoReichardt17automorphisms}, we give fault-tolerant procedures for applying arbitrary Clifford operations on these encoded qubits, also using only two extra qubits, and fault-tolerant universal operations with four extra qubits, $19$ total.  Substantial fault-tolerance tests can therefore be run on a quantum computer with fewer than twenty qubits.  

Our procedures here are based on adding ``flags" to the syndrome-extraction circuits in order to catch the faults that can lead to correlated errors on the data.  Figure~\ref{f:513flaggedsyndromeextraction} shows an example.  Provided that syndromes are extracted in a careful order, detecting the possible presence of a correlated error is enough to correct it.  

\begin{figure}[b]
\centering
\includegraphics[scale=.455]{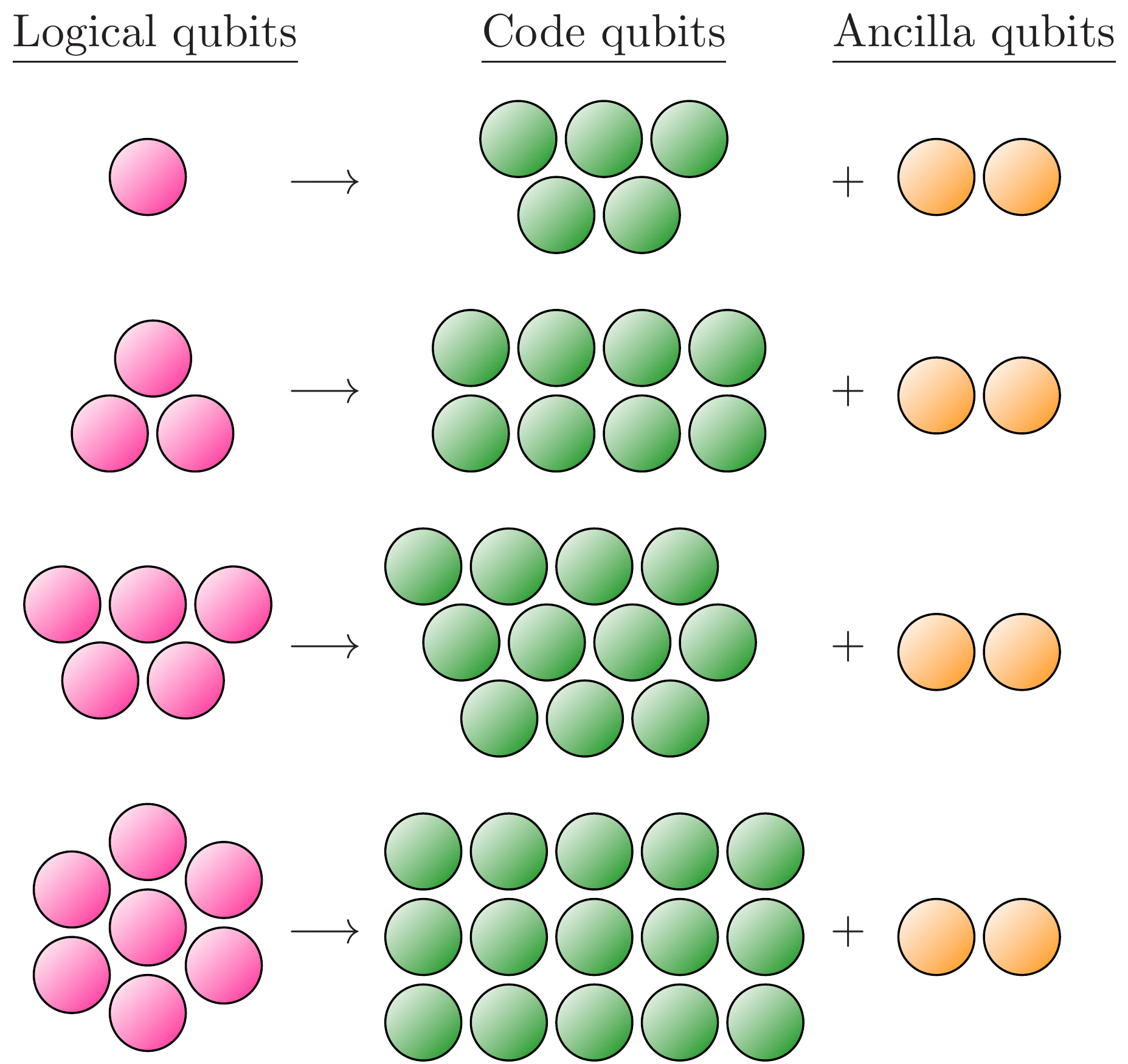}
\caption{Our flagged syndrome-extraction method uses just two ancilla qubits to correct errors fault tolerantly, for a variety of distance-three codes.  For example, seven qubits total are enough to correct errors on one encoded qubit, using the $\llbracket 5,1,3 \rrbracket$ code, and $17$ qubits are enough to protect seven encoded qubits.} \label{f:ecqubits}
\end{figure}

\begin{table}
\begin{center}
\begin{tabular}{c|c@{$\qquad$}c@{$\;\;\;\;$}c}
\hline \hline
& \multicolumn{3}{c}{Ancilla qubits required for} \\
 & {Shor} & {Decoded} & \\[-.1cm]
Code & {cat state} & {half cat} & Flagged \\
\hline
$\llbracket 5,1,3 \rrbracket$ & 5 & 3 & 2 \\
$\!\!\!\!\diamond\!$ $\llbracket 7,1,3 \rrbracket$ & 5 & 3 & 2~\cite{YoderKim16trianglecodes} \hspace{-.55cm} \\
$\llbracket 9,1,3 \rrbracket$ & 1 & -- & -- \\
$\llbracket 8,3,3 \rrbracket$ & 7 & 3 & 2 \\
$\llbracket 10,4,3 \rrbracket$ & 9 & 4 & 2 \\
$\llbracket 11,5,3 \rrbracket$ & 9 & 4 & 2 \\
$\!\!\!\!\diamond\!$ $\llbracket 15,7,3 \rrbracket$ & 9 & 4 & 2 \\
$\!\!\!\!\diamond\!$ $\llbracket 31,21,3 \rrbracket$ & 17 & 8 & 2 \\
$\!\!\!\!\diamond\!$ $\llbracket 2^r - 1, 2^r - 1 - 2r, 3 \rrbracket$ & $2^{r-1}+1$ & $2^{r-2}$ & 2 \\
\hline \hline
\end{tabular}
\end{center}
\caption{Numbers of ancilla qubits required for different fault-tolerant syndrome-extraction methods.  
Shor's method, using a verified cat state, requires $w + 1$ ancilla qubits, where $w$ is the largest weight of a stabilizer generator.  (The $\llbracket 9,1,3 \rrbracket$ Bacon-Shor code is an exception, as it is a subsystem code.)  Stephens, Yoder and Kim~\cite{Stephens14colorcodeft, YoderKim16trianglecodes} have proposed using a cat state on only $\max\{ 3, \lceil w/2 \rceil\}$ qubits, combined with a decoding trick from~\cite{DiVincenzoAliferis06slow}; see \appref{s:decodedhalfcatsyndromeextraction}.  When applicable, our flagged error-correction procedure needs only two extra qubits, as observed for the $\llbracket 7,1,3 \rrbracket$ code by~\cite{YoderKim16trianglecodes}.  Codes marked $\diamond$ are Hamming codes.} \label{f:ecancillassummary}
\end{table}

\newcommand{\circledletter}[1]{\mathop{\text{\textcircled{\raisebox{-.02cm}{$\scriptsize{#1}$}}}}\nolimits}

\begin{figure*}
\centering
\begin{tabular}{c@{$\;\;$}c@{$\!\!\!$}c@{$\;\;$}c}
\subfigure[\label{f:513code}]{\hspace{-.3cm}\raisebox{1.4cm}{
$\begin{array}{r c c c c c}
&X&Z&Z&X&I\\
&I&X&Z&Z&X\\
&X&I&X&Z&Z\\
&Z&X&I&X&Z\\
\cline{2-6}
\widebar X = &X&X&X&X&X\\
\widebar Z = &Z&Z&Z&Z&Z \\[.1cm]
\,
\end{array}
$}}
&
\subfigure[\label{f:513syndromeextraction}]{
\raisebox{0cm}{\includegraphics[scale=.769]{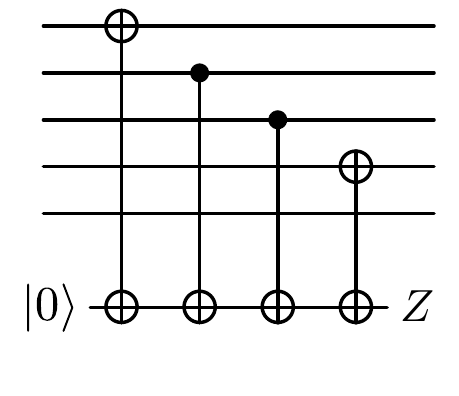}}
}
&
\subfigure[\label{f:513flaggedsyndromeextraction}]{
\raisebox{0cm}{\includegraphics[scale=.769]{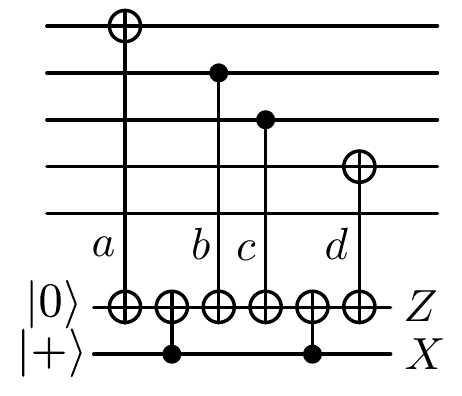}}
}
&
\subfigure[\label{f:513correlatederrors}]{
\raisebox{1.5cm}{
\begin{tabular}{c @{\;\;} c}
\hline \hline
$\circledletter{b}$ failure & Data error  \\
\hline
$IZ$ & $IIZXI$ \\
$XZ$ & $IXZXI$ \\
$YZ$ & $IYZXI$ \\
$ZZ$ & $IZZXI$ \\
\hline \hline \\
\end{tabular}
$\,$
\begin{tabular}{c @{\;\;} c}
\hline \hline
$\circledletter{c}$ failure & Data error  \\
\hline
$IZ$ & $IIIXI$ \\
$XZ$ & $IIXXI$ \\
$YZ$ & $IIYXI$ \\
$ZZ$ & $IIZXI$ \\
\hline \hline \\
\end{tabular}
}
}
\end{tabular}
\caption{
Flagged syndrome extraction for the $\llbracket 5,1,3 \rrbracket$ code.  
(a) Stabilizers and logical operators.  
(b) Circuit to extract the syndrome of the $XZZXI$ stabilizer into an ancilla qubit.  This is not fault tolerant because a fault on the ancilla can spread to a weight-two error on the data.  
(c) This circuit also extracts the $XZZXI$ syndrome, and if a single fault spreads to a data error of weight $\geq 2$ then the $X$ measurement will return $\ket -$, flagging the failure.  
(d) The nontrivial data errors that can result from a single gate failure in (c) that triggers the flag; these errors are distinguishable by their syndromes and so can be corrected.  
}
\end{figure*}

\section{Two-qubit fault-tolerant error correction for the $\llbracket 5,1,3 \rrbracket$ code}

The perfect $\llbracket 5,1,3 \rrbracket$ code~\cite{LaflammeMiquelPazZurek96} has the stabilizer generators, and logical $X$ and $Z$ operators of \figref{f:513code}.  (For the basics of stabilizer algebra, quantum error-correcting codes and fault-tolerant quantum computation, we refer the reader to~\cite{Gottesman09faulttolerance}.)  

The syndrome for the first stabilizer, $XZZXI$, can be extracted by the circuit in \figref{f:513syndromeextraction}, where $Z$ indicates a $\ket 0, \ket 1$ measurement, and \raisebox{-.1cm}{\includegraphics[scale=.45]{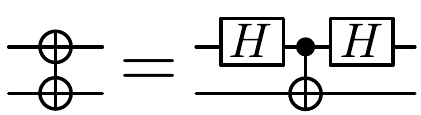}}$\;$.  However, this circuit is not fault tolerant.  For example, if the second gate fails and after it is applied an $IZ$ fault, then this fault will propagate through the subsequent gates to become an $IIZXI$ error on the data.  Thus a single fault can create a weight-two error on the data, and even a perfect error-correction procedure will correct this error in the wrong direction, creating the logical error $IIZXZ \sim \widebar X$.  

To fix this problem, we instead extract the $XZZXI$ syndrome using the circuit in \figref{f:513flaggedsyndromeextraction}.  With no faults, this circuit behaves exactly the same as that of \figref{f:513syndromeextraction}, and the $X$-basis ($\ket +, \ket -$) measurement will always give~$\ket +$.  The purpose of the $\ket +$ ancilla qubit, which we term a ``flag," is to detect gate faults that can propagate to correlated errors, weight two or higher, on the data.  Failures after gates $a$ and $d$ cannot create correlated errors.  The failures after gates $b$ and $c$ that can create correlated errors are listed in \figref{f:513correlatederrors}.  ($Y$ failures on the second qubit have the same effect on the data as $Z$ failures.)  These failures will all be detected, causing a $\ket -$ measurement outcome.  Moreover, observe that the seven distinct data errors have distinct, nontrivial syndromes.  

\begin{figure*}
\centering
\begin{tabular}{c@{$\;$}c@{$\!\!\!$}c@{$\;$}c}
\subfigure[\label{f:713code}]{\hspace{-.3cm}\raisebox{0cm}{
$\begin{array}{r c c c c c c c}
&I&I&I&X&X&X&X\\
&I&X&X&I&I&X&X\\
&X&I&X&I&X&I&X\\
&I&I&I&Z&Z&Z&Z\\
&I&Z&Z&I&I&Z&Z\\
&Z&I&Z&I&Z&I&Z\\
\cline{2-8}
\widebar X = &X&X&X&X&X&X&X\\
\widebar Z = &Z&Z&Z&Z&Z&Z&Z
\end{array}$}}
&
\subfigure[\label{f:713syndromeextraction}]{
\raisebox{-1.73cm}{\includegraphics[scale=.769]{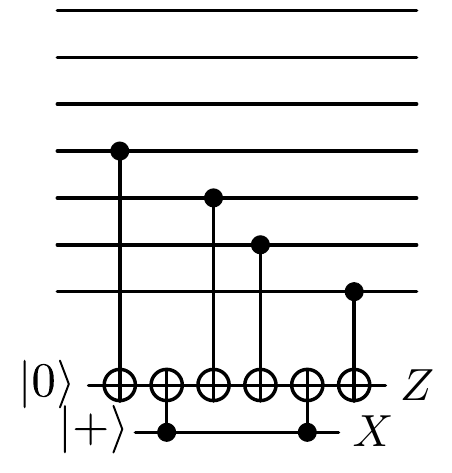}}
}
&
\subfigure[\label{f:1573code}]{
\raisebox{-.67cm}{
$\begin{tabular}{c c c c c c c c c c c c c c c}
0&0&0&0&0&0&0&1&1&1&1&1&1&1&1\\
0&0&0&1&1&1&1&0&0&0&0&1&1&1&1\\
0&1&1&0&0&1&1&0&0&1&1&0&0&1&1\\
1&0&1&0&1&0&1&0&1&0&1&0&1&0&1 \\
\\
\end{tabular}$
}}
&
\subfigure[\label{f:89101112131415ftsyndrome}]{
\raisebox{-1.73cm}{\includegraphics[scale=.769]{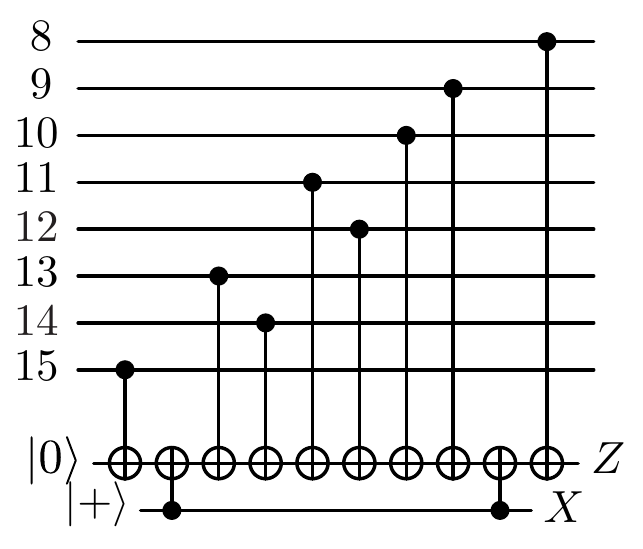}}
}
\end{tabular}
\caption{
Flagged syndrome extraction for Hamming codes.  
(a) $\llbracket 7,1,3 \rrbracket$ code stabilizers and logical operators.  
(b) Two-ancilla circuit to extract the $IIIZZZZ$ syndrome fault tolerantly.  
(c) $\llbracket 15,7,3 \rrbracket$ code parity checks.  
(d) Flagged circuit to extract the $Z_{\{8,\ldots,15\}}$ syndrome fault tolerantly.  
}
\end{figure*}

The complete error-correction procedure is given by: 

{ \noindent \hrulefill \\
\centering \textbf{Error-correction procedure} \\ } \smallskip
\begin{enumerate}[leftmargin=*]
\item Use the circuit of \figref{f:513flaggedsyndromeextraction} to extract the $XZZXI$ syndrome.  
\begin{enumerate}[leftmargin=.5cm]
\item If the flag qubit is measured as $\ket -$, then use the unflagged circuits analogous to \figref{f:513syndromeextraction} to extract all four syndromes.  Finish by applying the corresponding correction from among $IIIII$, $IIZXI$, $IXZXI$, $IYZXI$, $IZZXI$, $IIIXI$, $IIXXI$, $IIYXI$.  
\item Otherwise, if the syndrome is $-1$, i.e., the syndrome qubit is measured as $\ket 1$, then use unflagged circuits to extract \emph{all four} syndromes.  Finish by applying the corresponding correction of weight $\leq 1$.  
\end{enumerate}
\item (If the flag was not raised and the syndrome was trivial, then) Similarly extract the $IXZZX$ syndrome.  If the flag is raised, then use unflagged circuits to extract the four syndromes, and finish by applying the correction from among $IIIII$, $IIIIX$, $IXXII$, $IIIXX$, $XIIIY$, $IXIII$, $IIIZX$, $IIIYX$.  If the syndrome is nontrivial, then use unflagged circuits to extract the four syndromes, and finish by applying the correction of weight $\leq 1$.  
\item Similarly extract the $XIXZZ$ syndrome, and correct if the flag is raised or the syndrome is nontrivial.  
\item Similarly extract the $ZXIXZ$ syndrome, and correct if the flag is raised or the syndrome is nontrivial.  
\end{enumerate}
\vspace{-1\baselineskip}
\hrulefill
\medskip

We now argue that this procedure is fault tolerant.  
\begin{itemize}[leftmargin=*]
\item 
If there are no faults, then it appropriately corrects the data to the codespace.  
\item 
If the data lies in the codespace and there is at most one faulty gate in the error-correction procedure, then: 
\begin{itemize}[leftmargin=.5cm]
\item 
If all syndromes and flags are trivial, then the data can have at most a weight-one error.  (No correction is applied.)  
\item 
If a flag is raised or a syndrome is nontrivial, then the subsequent unflagged syndrome extractions are perfect, and suffice to correct either a possibly correlated error (if the flag is raised) or a weight $\leq 1$ error (if no flag is raised but the syndrome is nontrivial).  
\end{itemize}
\end{itemize}
Let us point out that when a syndrome extracted by a flagged circuit is nontrivial, then even if the flag is not raised we still extract all four syndromes (using unflagged circuits) before applying a correction.  This is because a fault on the data could have been introduced in the middle of syndrome extraction.  For example, if we extract the first syndrome as $+1$, but a $Z_1$ error is then added to the data, the remaining syndromes will be $+1, -1, +1$.  The correction for $(+, +, -, +)$ is $Z_3$, but were we to apply this correction the data would end up with error $ZIZII$.  This moral is that nontrivial syndromes cannot be trusted unless they are repeated.  

We give two-ancilla-qubit fault-tolerant state preparation and measurement circuits in \appref{s:513preparationmeasurement}.

\section{Two-qubit fault-tolerant error correction for Hamming codes}

The Hamming codes are a family of $\llbracket 2^r - 1, 2^r - 1 - 2r, 3 \rrbracket$ quantum error-correcting codes, for $r = 3, 4, 5, \ldots$.  They are self-dual perfect CSS codes.

\subsection{$\llbracket 7,1,3 \rrbracket$ Steane code}

The $\llbracket 7,1,3 \rrbracket$ code, known as Steane's code~\cite{Steane96css}, has the stabilizers and logical operators given in \figref{f:713code}.  

\begin{figure*}
\centering
\begin{tabular}{c@{$\quad\qquad$}c@{$\quad\qquad$}c}
\subfigure[]{\raisebox{-.23cm}{
$\begin{array}{r c c c c c c c c}
&X&X&X&X&X&X&X&X\\
&Z&Z&Z&Z&Z&Z&Z&Z\\
&I&I&Z&Y&X&Z&Y&X\\
&I&Z&X&I&X&Y&Z&Y\\
&I&X&I&Z&Z&X&Y&Y
\end{array}$}}
&
\subfigure[]{
$\begin{array}{r c c c c c c c c c c}
&X&X&X&X&X&X&X&X&X&X\\
&Z&Z&Z&Z&Z&Z&Z&Z&Z&Z\\
&X&Z&Z&X&I&Z&Y&Y&Z&I\\
&I&X&Z&Z&X&I&Z&Y&Y&Z\\
&X&I&X&Z&Z&Z&I&Z&Y&Y\\
&Z&X&I&X&Z&Y&Z&I&Z&Y
\end{array}$}
&
\subfigure[]{
$\begin{array}{r c c c c c c c c c c c}
&I&I&I&Z&Z&Y&Y&X&X&X&X\\
&I&I&I&X&X&Z&Z&Y&Y&Y&Y\\
&I&X&X&Z&Z&I&I&Y&Y&Y&Y\\
&I&Z&Z&Z&Z&Z&Z&Z&Z&Z&Z\\
&Z&X&I&I&Z&Z&X&I&Y&X&Z\\
&Y&I&Z&I&Y&Z&Y&Y&Z&X&I
\end{array}$}
\end{tabular}
\caption{
(a) Stabilizers for an $\llbracket 8,3,3 \rrbracket$ code.  One choice for the logical operators is $\protect\widebar X_1 = X_{\{4,5,7,8\}}$, $\protect\widebar Z_1 = Z_{\{2,5\}} X_{\{3,8\}}$, $\protect\widebar X_2 = X_{\{3,6\}} Z_{\{4,5\}}$, $\protect\widebar Z_2 = Z_{\{1,5,6,7\}}$, $\protect\widebar X_3 = Z_{\{1,2\}} X_{\{6,7\}}$, $\protect\widebar Z_3 = Z_{\{1,2,4,7\}}$.  
(b) Stabilizers for a $\llbracket 10,4,3 \rrbracket$ code, based on the $\llbracket 5,1,3 \rrbracket$ code~\cite{CalderbankRainsShorSloane96smallcodes}.  One choice for the logical operators is $\protect\widebar X_1 = X_8 Y_{10} Z_1$, $\protect\widebar Z_1 = X_{10} Y_7 Z_3$, $\protect\widebar X_2 = X_{10} Y_8 Z_2$, $\protect\widebar Z_2 = X_8 Y_6 Z_5$, $\protect\widebar X_3 = X_9 Y_7 Z_1$, $\protect\widebar Z_3 = X_4 Y_3 Z_7$, $\protect\widebar X_4 = X_9 Y_6 Z_2$, $\protect\widebar Z_4 = X_4 Y_5 Z_6$.  
(c) Stabilizers for an $\llbracket 11, 5, 3 \rrbracket$ code.  
} \label{f:othercodes}
\end{figure*}

The circuit in \figref{f:713syndromeextraction} extracts the $IIIZZZZ$ syndrome.  As in \figref{f:513syndromeextraction}, any single gate fault that leads to a data error of weight $\geq 2$ will also make the $X$ measurement output $\ket -$.  Moreover, if that measurement gives $\ket -$ with a single fault, the error's possible $Z$ components~are 
\begin{equation*}
\identity, \; Z_7, \; Z_6 Z_7, \; Z_5 Z_6 Z_7 \sim Z_4
\end{equation*}
These errors are distinguishable by their $X$ syndromes.  

Similar circuits work for the other stabilizers, thus giving a two-qubit fault-tolerant error-correction procedure.  

\subsection{$\llbracket 15,7,3 \rrbracket$ Hamming code}

The $r = 4$, $\llbracket 15,7,3 \rrbracket$ code has four $X$ and four $Z$ stabilizers each given by the parity checks of \figref{f:1573code}.  Observe that the columns form the numbers $1$ through $2^r - 1$, written in binary; this holds for each Hamming code.  

For a subset~$S$, let $Z_S = \prod_{j \in S} Z_j$.  The circuit of \figref{f:89101112131415ftsyndrome} 
works to fault-tolerantly extract the first $Z$ syndrome, for $Z_{\{8, \ldots, 15\}}$.  If the flag is triggered by a single fault, then the possible error $Z$ components are 
\begin{equation*}\begin{gathered}
\identity, Z_8, Z_{\{8, 9\}}, Z_{\{8, 9, 10\}}, Z_{\{8, 9, 10, 12\}}, Z_{\{8, 9, 10, 11, 12\}}, \\
Z_{\{8, 9, 10, 11, 12, 14\}}, 
Z_{\{8, 9, 10, 11, 12, 13, 14\}}
\end{gathered}\end{equation*}
Once again, these errors are distinguishable by their $X$ syndromes.  (The order of the CNOT gates is important.)  

\appref{s:1573statepreparation} gives a single-ancilla fault-tolerant circuit to prepare encoded $\ket{0^7}$.

\subsection{General Hamming codes}

\begin{claim} \label{t:hamming}
Syndromes for the $\llbracket 2^r - 1, 2^r - 1 - 2r, 3 \rrbracket$ Hamming code can be fault-tolerantly extracted with two qubits.  
\end{claim}

\begin{proof}
Consider the stabilizer $Z_{\{2^{r-1}, \ldots, 2^r - 1\}}$.  As in Figs.~\ref{f:713syndromeextraction} and~\ref{f:89101112131415ftsyndrome}, we give a permutation of the last $2^{r-1}$ qubits so that, when the flag is triggered, the possible $Z$ errors have distinct syndromes.  

Let $p(x)$ be a degree-$(r-1)$ primitive polynomial over $\mathrm{GF}(2)$.  For $j = 0, \ldots, 2^{r-1} - 2$, let $q_j(x) = x^j \mod p$; these are distinct polynomials of degree $\leq r-2$.  Furthermore, the remainders of their cumulative sums, $\sum_{i=0}^j x^i$, are also distinct.  (Otherwise, if $0 \equiv x^{j+1} + \ldots + x^k = x^{j+1} (x^{k-j} + 1) / (x + 1)$, then considering the lowest-order terms give $x^{k-j} + 1 \equiv 0$, contradicting that $p$ is primitive.)  

The desired permutation is $2^{r-1}, q_0(2) + 2^{r-1}, q_1(2) + 2^{r-1}, \ldots, q_{2^{r-1} - 2}(2) + 2^{r-1}$.  
\end{proof}

\section{Two-qubit fault-tolerant \\ error correction for \\ other \mbox{distance-three codes}}

Figure~\ref{f:othercodes} presents $\llbracket 8,3,3 \rrbracket$, $\llbracket 10,4,3 \rrbracket$ and $\llbracket 11,5,3 \rrbracket$ codes~\cite{Gottesman96codes, Steane96codes, CalderbankRainsShorSloane96smallcodes}.  Many other distance-three codes are given in the code tables at~\cite{Grassl07codetable}.  

The general property required for our flagged procedure to work for measuring an operator is: for an operator $Z_1 Z_2 Z_3 \ldots Z_w$ (up to qubit permutations and local Clifford unitaries), the different errors $P_j Z_{j+1} \ldots Z_w$, for $P \in \{I, X, Y, Z\}$ and $j \in \{2, \ldots, w-1\}$, should have distinct and nontrivial syndromes.  (For a CSS code, for which $X$ and $Z$ errors can be corrected separately, it is enough that the syndromes for different errors be different for $P \in \{I, Z\}$.)  

We have verified this property for some qubit permutation of each of the stabilizer generators for the $\llbracket 10,4,3 \rrbracket$ and $\llbracket 11,5,3 \rrbracket$ codes from \figref{f:othercodes}.  For the $\llbracket 8,3,3 \rrbracket$ code, we have found permutations that work for each of the stabilizer generators except $X^{\otimes 8}$ and $Z^{\otimes 8}$.  However, we can replace these stabilizer generators with, respectively, $XXYZIYZI$ and $ZZIXYIXY$, which do satisfy the desired property.  For example, to extract the syndrome of $XXYZIYZI$, couple the qubits to the ancilla in the order $1, 2, 3, 6, 4, 7$.  (With the order $1, 2, 3, 4, 6, 7$, the errors $IIIIIYZI$ and $IXYZIYZI$ have the same syndrome.)  

Therefore, for all of these codes, two ancilla qubits are enough to fault-tolerantly extract the syndromes and apply error correction.

\section{Two-qubit fault-tolerant error detection for $\llbracket n,n-2,2 \rrbracket$ codes}

The idea of flagging faults that can spread badly is also useful for error-detecting codes.  

For even $n$, the $\llbracket n, n-2, 2 \rrbracket$ error-detecting code has stabilizers $X^{\otimes n}$ and $Z^{\otimes n}$, and logical operators $\widebar{X}_j = X_1 X_{j+1}$, $\widebar Z_j = Z_{j+1} Z_n$ for $j = 1, \ldots, n-2$.  

Observe that extracting a syndrome with a single ancilla qubit is not fault tolerant because, for example, a $Z$ fault at either location indicated with a ${\color{red} \bigstar}$ in \figref{f:6qubitextraction} results in an undetectable logical error.  With a flag qubit, the circuit is fault tolerant; any single fault is either detectable or creates no data error.  

\begin{figure}
\centering
\begin{tabular}{c@{$\!$}c}
\subfigure[\label{f:6qubitextraction}]{
\hspace{-.4cm}
\raisebox{-1.5cm}{\includegraphics[scale=.769]{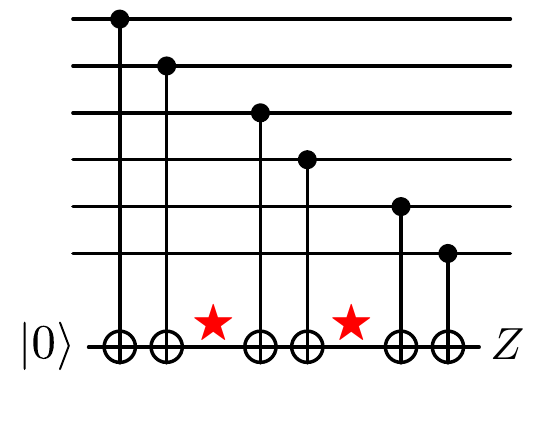}}
}
&
\subfigure[\label{f:6qubitflaggedextraction}]{
\raisebox{-1.5cm}{\includegraphics[scale=.769]{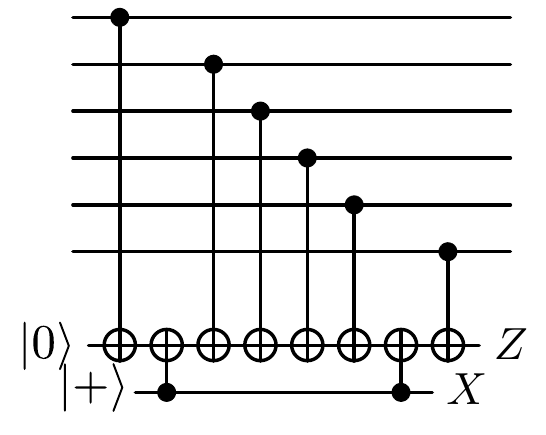}}
}
\end{tabular}
\caption{
(a) Circuit to extract the $Z^{\otimes n}$ syndrome. 
(b) Adding a flag makes it fault tolerant.   
}
\end{figure}

This approach works for any~$n$.  One way of interpreting it is that the ancilla is encoded into the two-qubit $Z$-error detecting code, with stabilizers $XX$ and logical operators $\widebar X = XI$ and $\widebar Z = ZZ$.  This code detects the single $Z$ faults that can propagate back to the data.  

In \appref{s:errordetectionpreparationmeasurement} we give single-ancilla fault-tolerant circuits for initialization and projective measurement.

\section{Simulations and conclusion}

In order to get a sense for the practicality of the two-qubit error-correction schemes, we simulate error correction using a standard depolarizing noise error model on the one- and two-qubit operations~\cite{Knill05}.\footnote{The source code is at \url{https://github.com/chaor11/twoqubitec}.}  
Figure~\ref{f:simulationec} shows the results from simulating at least $10^6$ consecutive error-correction rounds for each value of the CNOT failure rate~$p$.  Observe that for the $\llbracket 15,7,3 \rrbracket$ code, Steane-style error correction, which extracts all stabilizer syndromes at once, performs better than either the Shor-style or two-qubit procedures.  For the $\llbracket 5,1,3 \rrbracket$ and $\llbracket 7,1,3 \rrbracket$ codes, the different error-correction methods are all very close.  Note that we do not introduce memory errors on resting qubits, and nor do we add errors for moving qubits into position to apply the gates; we leave more detailed simulations and optimizations to future work. 
 
\begin{figure}
\centering
\hspace{-.2cm}
\includegraphics[scale=.46]{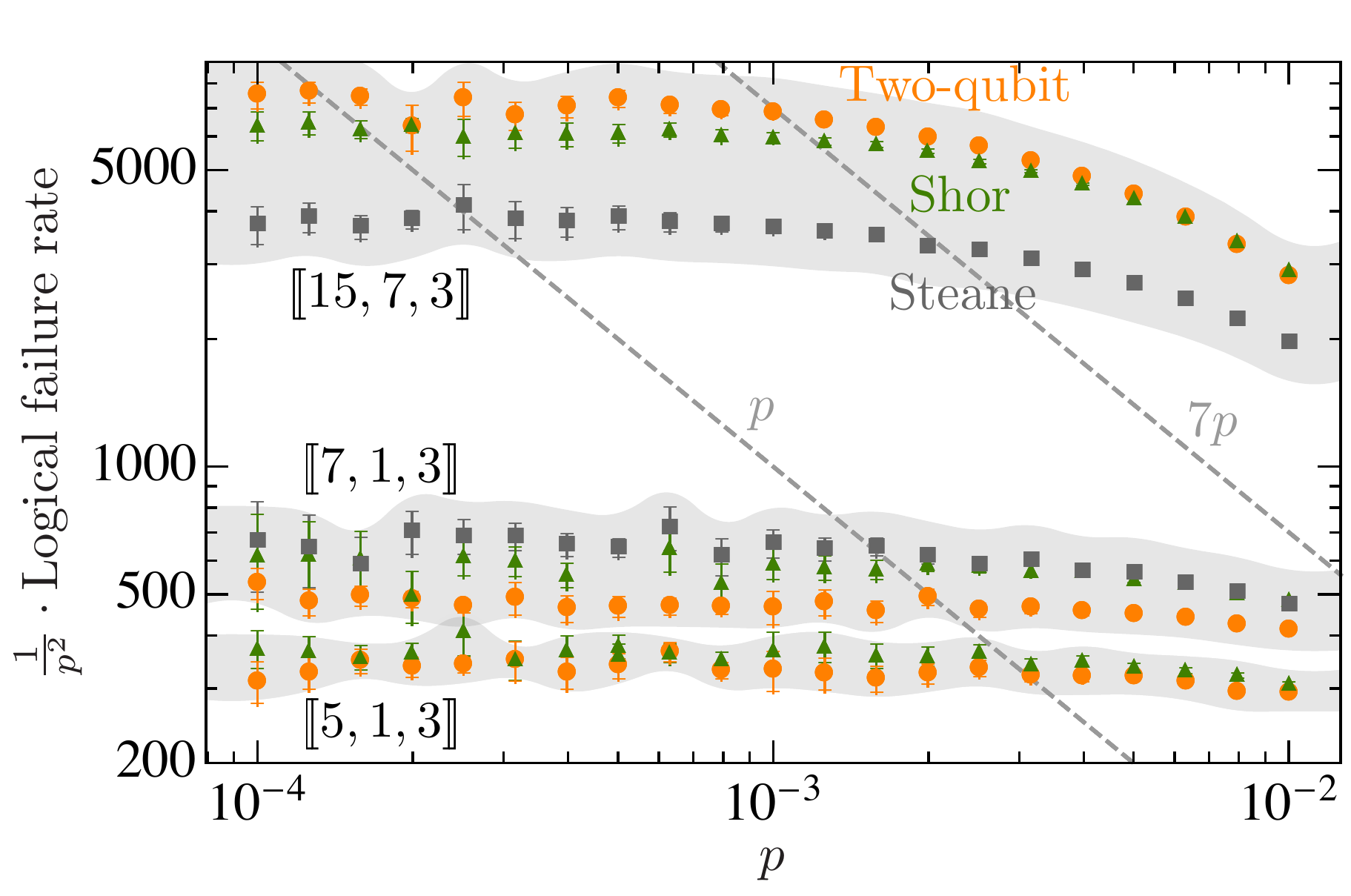}
\caption{Comparison to previous error-correction schemes.  Two-qubit error correction (\protect\scalebox{3.5}{\raisebox{-.05cm}{\color{Tangerine}$\hspace{-.02cm}\cdot\hspace{-.02cm}$}}) slightly outperforms Shor-style correction~(\protect\scalebox{.6}{\raisebox{.05cm}{\color{Fern}$\hspace{-.02cm}\blacktriangle\hspace{-.02cm}$}}) for the $\llbracket 5,1,3 \rrbracket$ code, and both Shor- and Steane-style (\protect\scalebox{.5}{\raisebox{.05cm}{\color{Steel}$\hspace{-.02cm}\blacksquare\hspace{-.02cm}$}}) error correction for the $\llbracket 7,1,3 \rrbracket$ code.  For the $\llbracket 15,7,3 \rrbracket$ code, Steane's method performs best.  Logical error rates are plotted divided by $p^2$ to reveal leading-order coefficients, and $p$ and $7 p$ are plotted to help judge pseudo-thresholds.} \label{f:simulationec}
\end{figure}

We have focused on extracting syndromes using two ancilla qubits in order to minimize the qubit overhead.  With just one more qubit, however, fault-tolerant syndrome extraction becomes considerably simpler.  Consider for example the circuit in \figref{f:weight42flags} for extracting a $ZZZZ$ syndrome.  

\begin{figure}
\raisebox{-1.7cm}{\includegraphics[scale=.769]{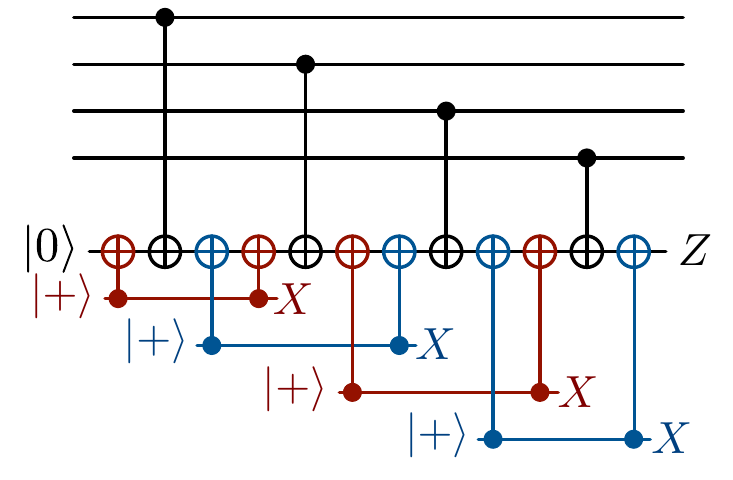}}
\caption{Overlapping flags can closely localize $Z$ faults.} \label{f:weight42flags}
\end{figure}

Every CNOT gate into the syndrome qubit has its own flag, to catch $Z$ faults that can lead to correlated $Z$ errors on the data.  The flags allow for closely localizing any fault, thereby easing error recovery.  For example, if there is a single fault and only the second flag above is triggered, then the $Z$ error on the data can be $IIII$, $IZZZ$ or $IIZZ$.  The regions covered by the flags overlap so that no $Z$ fault is missed.  This technique is most effective if qubit initialization and measurement is fast.  Less extreme versions of the technique, in which some gates share flags, can also be used.  In~\cite{ChaoReichardt17automorphisms} we use variants of this technique to apply operations fault tolerantly to data in a block code, with little qubit overhead.  

\begin{figure}
\centering
\begin{tabular}{c}
\subfigure[\label{f:713flagextract2syndromes}]{\raisebox{-1.7cm}{\includegraphics[scale=.769]{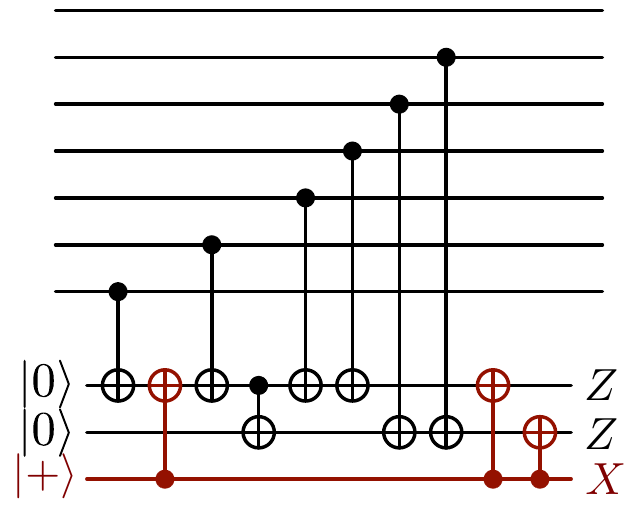}}} \\
\subfigure[\label{f:713flagextract3syndromes}]{\raisebox{-1.7cm}{\includegraphics[scale=.769]{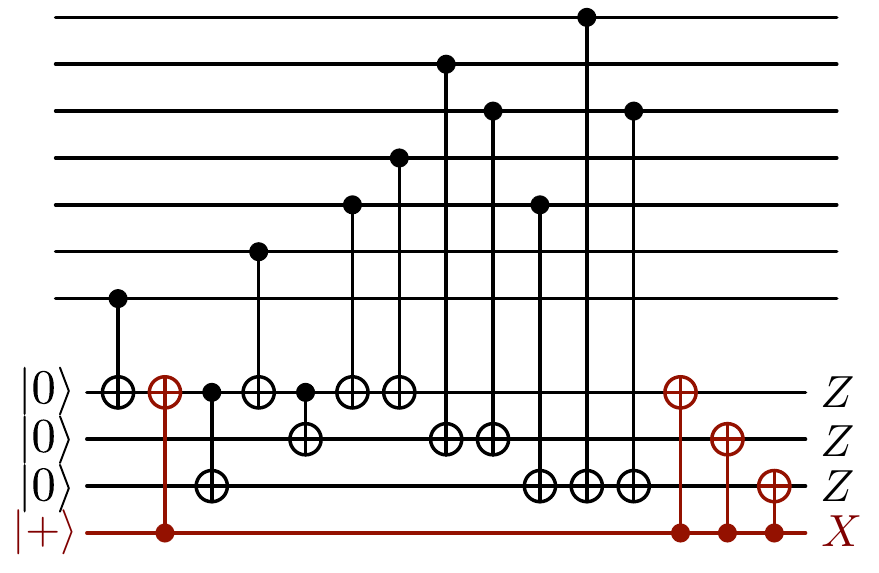}}}
\end{tabular}
\caption{Extracting multiple syndromes with a shared flag.}
\end{figure}

\begin{figure*}
\centering
\begin{tabular}{c@{$\;$}c@{$\;$}c}
\subfigure[\label{f:halfdasyndrome6}]{
\raisebox{-1.73cm}{\includegraphics[scale=.769]{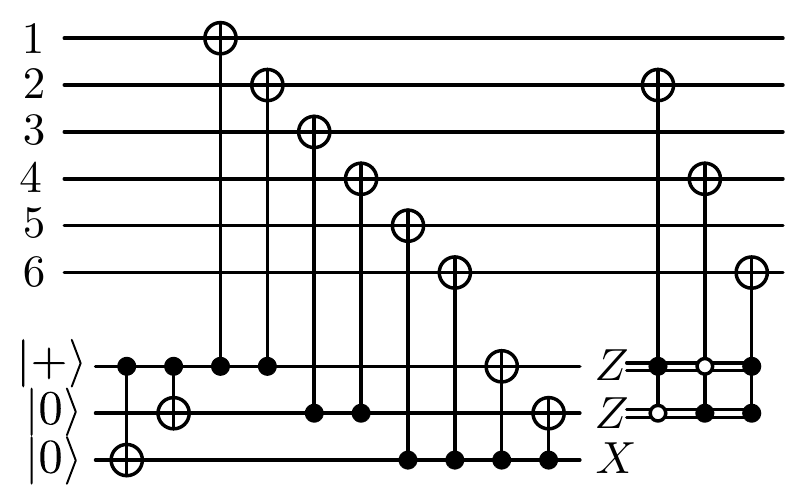}}
}
&
\subfigure[\label{f:halfdasyndrome10}]{
\raisebox{-1.73cm}{\includegraphics[scale=.769]{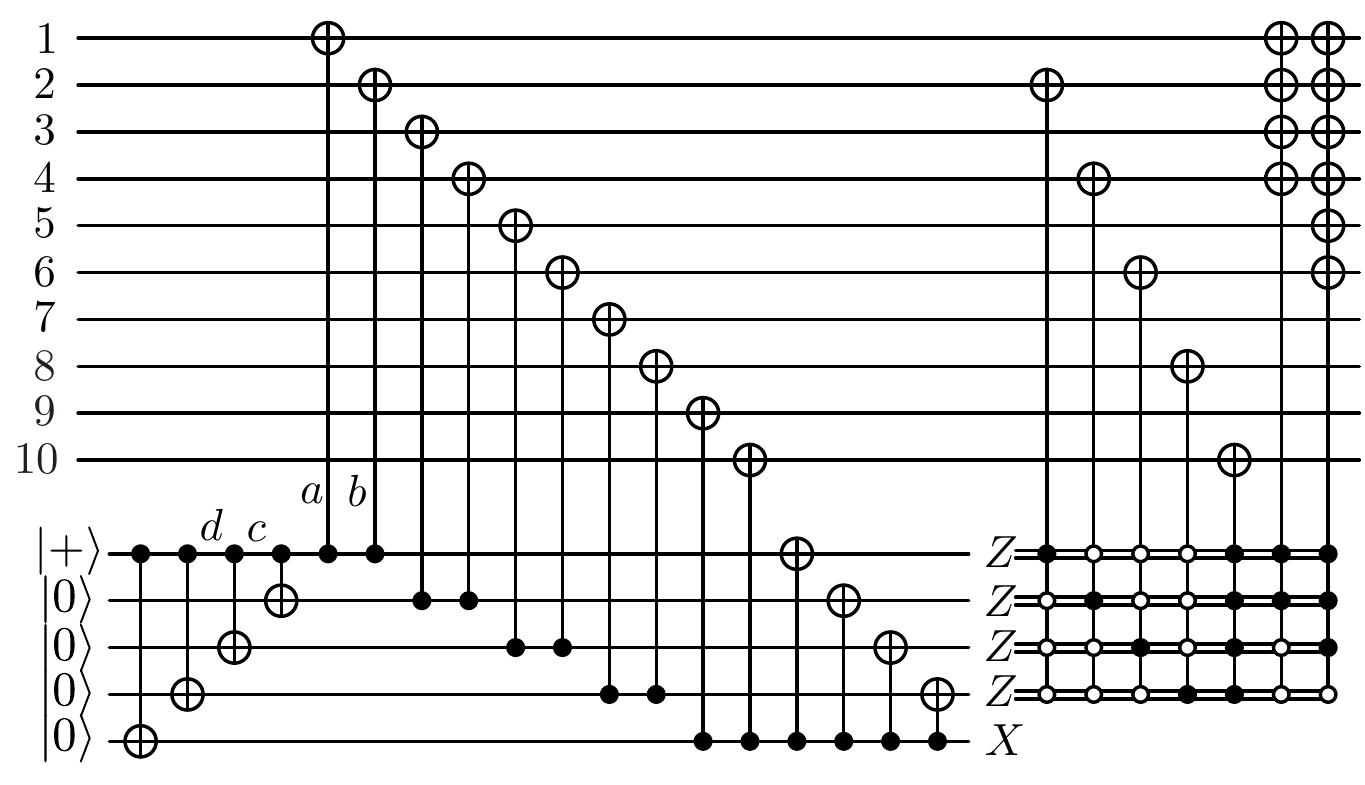}}
}
\end{tabular}
\caption{Stephens-Yoder-Kim syndrome extraction, for a distance-three code.  Up to two data qubits are coupled to each ancilla.}
\end{figure*}

For a large code with many stabilizers, it can be inefficient to extract the syndromes one at a time.  The circuit in \figref{f:713flagextract2syndromes} extracts two of the $\llbracket 7,1,3 \rrbracket$ code's syndromes at once, using a shared flag.  Figure~\ref{f:713flagextract3syndromes} uses a shared flag to extract all three $Z$ syndromes.  A single fault can lead to at most a weight-one $X$ error and, if the flag is not triggered, a weight-one $Z$ error.  If the flag is triggered, the gates are arranged so that the different $Z$ errors are distinguishable.  The circuit uses four qubits and $15$ CNOT gates, versus seven qubits and $25$ CNOT gates for Steane-style extraction with ancilla decoding~\cite{DiVincenzoAliferis06slow}; however the syndrome $001$ can occur with errors $\identity$, $Z_1$, $Z_3$ or $Z_5$ and so, unlike in Steane's scheme, must be verified before applying a correction.  

The design space for fault-tolerant error correction thus expands considerably with more allowed qubits.  A natural problem is to extend the flag technique to medium-size codes of higher distance.

\subsection*{Acknowledgements}

We thank Qian Yu for suggesting the construction of \claimref{t:hamming}, and thank Ted Yoder, Earl Campbell and James Wootton for helpful comments.  Research supported by NSF grant CCF-1254119 and ARO grant W911NF-12-1-0541.  

\appendix

\section{Stephens-Yoder-Kim space-optimized syndrome extraction} \label{s:decodedhalfcatsyndromeextraction}

For a distance-three code, Shor's method~\cite{Shor96} uses $w + 1$ ancilla qubits to fault-tolerantly extract the syndrome of a weight-$w$ stabilizer, by preparing and verifying a $w$-qubit cat state $\tfrac{1}{\sqrt 2}(\ket{0^w} + \ket{1^w})$.  For example, to extract an $XXXX$ syndrome, $w = 4$, run the following circuit:  
\begin{equation*}
\raisebox{-1.2cm}{\includegraphics[scale=.769]{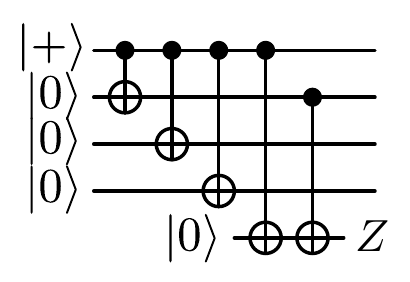}}
\end{equation*}
The $Z$ measurement is to catch correlated $X$ errors.  If it returns $\ket 0$, apply four CNOT gates into the data qubits and measure each ancilla in the $X$ basis.  

For $w = 4$ and $w = 7$, DiVincenzo and Aliferis~\cite{DiVincenzoAliferis06slow} give circuits that use unverified $w$-qubit cat states, that interact with the data without waiting for a measurement to finish.  Instead of measuring the cat state immediately after coupling it to the data, they first apply a carefully designed decoding circuit.  They conjecture the technique extends to arbitrary $w$.  

In fact, one can do better.  Stephens~\cite{Stephens14colorcodeft} has proposed that for $w \leq 8$, four ancilla qubits are enough to fault-tolerantly extract the syndrome, and Yoder and Kim~\cite{YoderKim16trianglecodes} have shown that three qubits are enough for $w = 4$.  We will show that in general, for a distance-three code, $\max\{3, \lceil w/2 \rceil\}$ qubits suffice to fault-tolerantly extract a weight-$w$ stabilizer's syndrome.  Unlike our two-qubit schemes, this method can also be ``deterministic"~\cite{CrossDiVincenzoTerhal09codes}.  

The basic idea is to follow the DiVincenzo-Aliferis method, except coupling two data qubits to each ancilla qubit; with appropriate corrections, this remains fault tolerant.  Figure~\ref{f:halfdasyndrome6} sketches the technique for the minimum number, three, of ancilla qubits, and \figref{f:halfdasyndrome10} shows the construction for $w = 10$.  The general case follows a similar pattern.  

\begin{figure*}
\centering
\begin{tabular}{c@{$\quad$}c}
\subfigure[\label{f:513xpreparation}]{
\begin{minipage}[b][1.63in]{2.82in}  
\raisebox{.115cm}{\includegraphics[scale=.769]{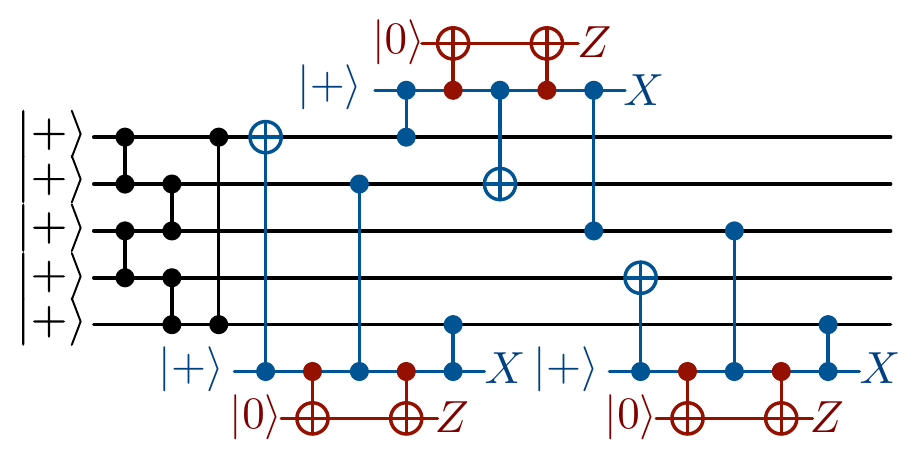}}
\end{minipage}
}
&
\subfigure[\label{f:513xmeasurement}]{
\begin{minipage}[b][1.63in]{3.97in}
\raisebox{0cm}{\includegraphics[scale=.769]{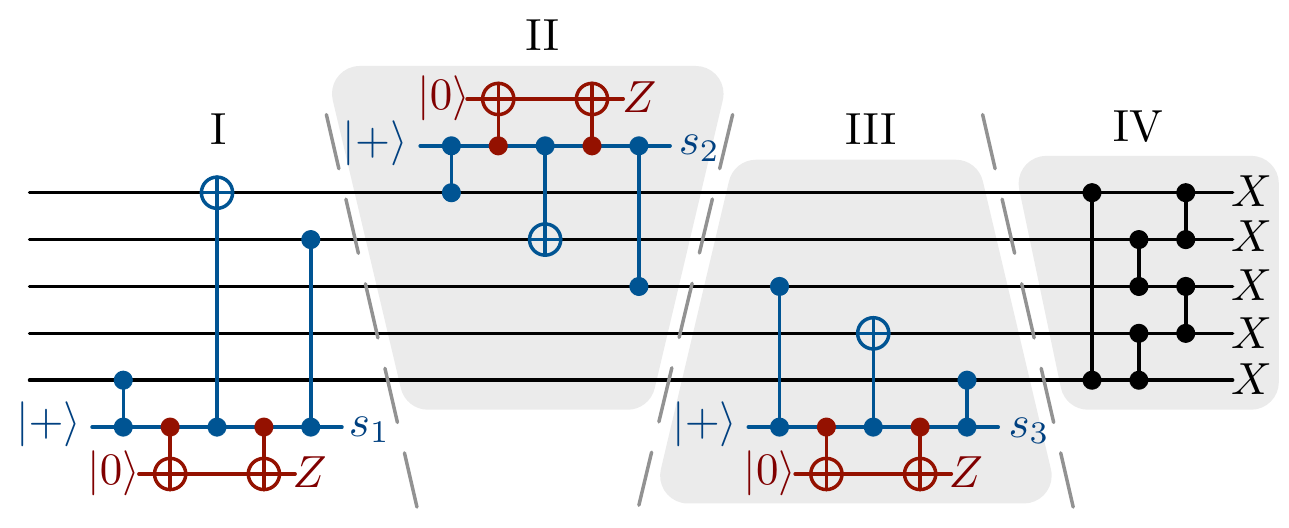}}
\end{minipage}
}
\end{tabular}
\caption{Fault-tolerant circuits to (a) prepare encoded $\ket +$, and (b) measure logical $X$, for the $\llbracket 5,1,3 \rrbracket$ code.  In (a) the measurement outcomes should all be trivial, $\ket +$ or $\ket 0$, for the state to be acceptable.  In (b), portions of the circuit may be skipped depending on the measurement results.}
\end{figure*}

Observe that in \figref{f:halfdasyndrome10} a failure of gate a can propagate to a data error $P_1 X_2$ for any $P \in \{I, X, Y, Z\}$.  A failure of gate b can create a $P_2$ data error.  The $X_2$ correction, applied when the $Z$ measurements output $1000$, ensures that at the end the data has an error of weight $\leq 1$.  The corrections for $Z$ measurements $0100$, $0010$, $0001$ and $1111$ are similar.  Failures at locations $c$ and~$d$ can cause data errors of weight four or six, respectively, but these are caught and corrected for with $Z$ measurements $1100$ or $1110$.  It is important that after the coupling to the data no single fault can cause the measurements to give $1100$ or $1110$.

\section{Fault-tolerant state preparation and measurement for the $\llbracket 5, 1, 3 \rrbracket$ code} \label{s:513preparationmeasurement}

\subsection{State preparation}

The black portion of the circuit in \figref{f:513xpreparation} prepares, up to Paulis, encoded $\ket +$, the five-cycle graph state with stabilizers $ZXZII$ and its cyclic permutations.  Then three of the syndromes are measured.  Provided that there is at most one gate fault, and the syndrome (blue) and flag (red) measurements are trivial, the output has at most a weight-one error.  One can apply transversal Clifford operators to obtain encoded $\ket 0$.

\subsection{Measurement}

Although logical $Z$ is transversal, $\widebar X = XXXXX$, it is not fault tolerant to measure $X$ transversally.  The problem is that, as the $\llbracket 5,1,3 \rrbracket$ code is not CSS, the $X$ measurements do not give sufficient information to correct errors.  Fortunately, we can apply the same flag trick that we used for error correction.  We give a two-ancilla-qubit fault-tolerant circuit for destructively measuring $\widebar X$.  

\begin{enumerate}[leftmargin=*]
\item 
Let $s_1$ be the result of measuring $XZIIZ$ using a flag, as in part I of \figref{f:513xmeasurement}.  

If the flag is raised, skip to part IV.  Apply the five-$\CZ$ decoding circuit, measure the qubits and decode to get the answer.  (Given that the flag was raised, the data error can be $IIIII$, $IZIII$, $XZIII$, $YZIII$ or $ZZIII$ before decoding, and $IIIII, IZIII, XIIIZ, YIIIZ$ or $ZZIII$ after.  These errors can be distinguished and corrected for.)  

\item 
Measure $ZXZII$ with a flag, and call the result $s_2$, as in part II of \figref{f:513xmeasurement}.  

If the flag is raised, then return $s_1$.  

If the flag is not raised, and $s_2 \neq s_1$, then there has been a fault, which can mean a data error of weight $\leq 1$.  Skip to part IV: decode, measure and correct the result for the $16$ possible errors $IIIII, XIIII, \ldots, IIIIZ$.  

\item 
Measure $IIZXZ$ with a flag and call the result $s_3$, as in part III.  

If the flag is raised, then return $s_1 = s_2$.  

Otherwise, if $s_3 = s_1 = s_2$ then return this value.  

If $s_3 \neq s_1 = s_2$, then in part IV decode, measure and correct the result for the $16$ errors of weight $\leq 1$.  
\end{enumerate}

This procedure is fault tolerant.  For example, if the flag is raised while extracting $s_2$, then, provided the circuit has only one fault, $s_1$ must have been correct.

\section{Fault-tolerant state preparation \\ for the $\llbracket 15, 7, 3 \rrbracket$ code} \label{s:1573statepreparation}

Figure~\ref{f:1573encoding} gives a fault-tolerant circuit to prepare $\ket{0^7}$ encoded in the $\llbracket 15, 7, 3 \rrbracket$ code.  We have condensed CNOT gates with the same control wire.  The state is accepted provided that the measurement returns~$\ket 0$.  Note that since the code is perfect CSS, we need not worry about $Z$ faults propagating and only need to check for $X$ faults.  (Any $Z$ error is equivalent to either the identity or a weight-one error.)  

\begin{figure}
\raisebox{-3cm}{\includegraphics[scale=.769]{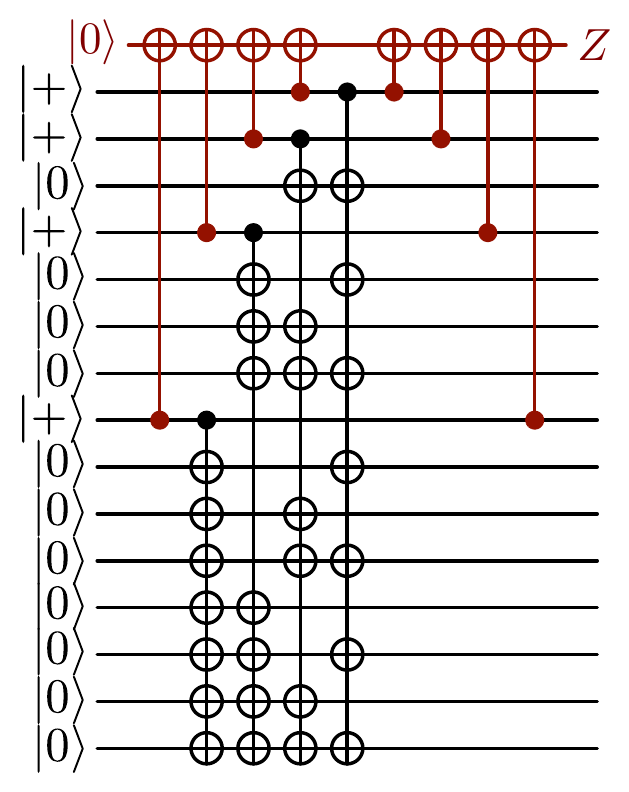}}
\caption{CSS fault-tolerant preparation of encoded $\ket{0^7}$.} \label{f:1573encoding}
\end{figure}

The circuit is fault tolerant only for $X$ and $Z$ faults considered separately.  For example, if the $\text{CNOT}_{1,15}$ gate fails as $ZX$, the state will be accepted with the weight-two error $Z_1 X_{15}$---but as the $X$ and $Z$ parts of this error each have weight one, the error is correctable.  

Other encoded computational basis states can be prepared by applying logical $X$ operators to $\ket{0^7}$.

\section{Fault-tolerant state preparation and measurement for the $\llbracket n, n-2, 2 \rrbracket$ codes} \label{s:errordetectionpreparationmeasurement}

We give one-qubit fault-tolerant circuits for state preparation and projective measurement, for the $\llbracket n, n-2, 2 \rrbracket$ codes.

\subsection{State preparation}

Encoded $\ket{0^{n-2}}$ is a cat state $\tfrac{1}{\sqrt 2}(\ket{0^n} + \ket{1^n})$.  It can be prepared fault tolerantly by 
\begin{equation} \label{e:prepare0}
\raisebox{-1.4cm}{\includegraphics[scale=.769]{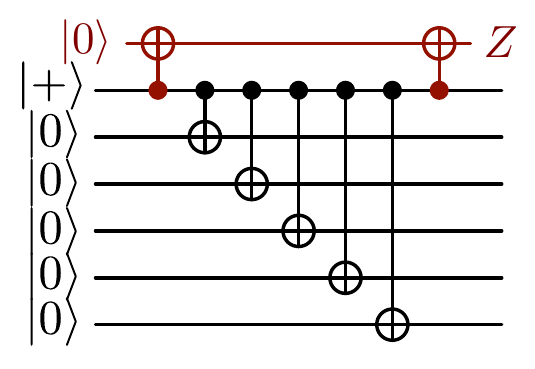}}
\end{equation}
where the $Z$ measurement should return~$\ket 0$.  

Encoded states $\ket{+^j 0^{n-2-j}}$ can be prepared by applying targeted logical Hadamard gates to encoded $\ket{0^{n-2}}$.  However, it is easier to prepare them directly.  
For $j$ odd, encoded $\ket{+^j 0^{n-2-j}}$ is $\tfrac12 (\ket{+^{j+1}} + \ket{-}^{j+1}) \otimes (\ket{0^{n-j-1}} + \ket{1^{n-j-1}})$, i.e., the tensor product of cat and dual cat states.  The states can be prepared separately as in~\eqnref{e:prepare0}.  The following circuit for $n = 6$, $j = 2$ generalizes naturally to any even~$j$: 
\begin{equation}
\raisebox{-1.4cm}{\includegraphics[scale=.769]{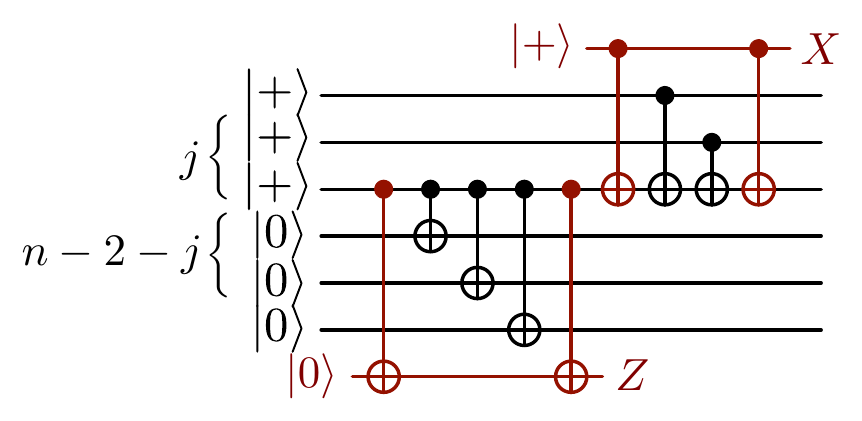}}\vspace{0cm}
\end{equation}
The black portion prepares the state, while the red parts catch faults using postselected $Z$ and $X$ measurements.

\subsection{Measurement}

The qubits can all be simultaneously be measured in the $X$ or $Z$ logical bases by transversal $X$ or $Z$ measurement.  To projectively measure just $\widebar Z_j$ fault tolerantly, measure it twice using one ancilla:  
\begin{equation}
\raisebox{-.7cm}{\includegraphics[scale=.769]{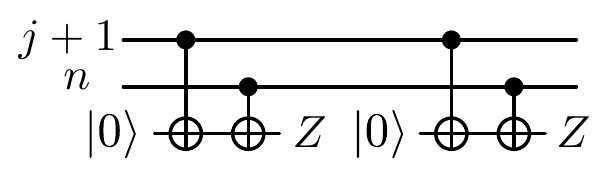}}
\end{equation}
Different results indicate an error.  
Similar circuits work for measuring $\widebar Z_i \widebar Z_j = Z_{i+1} Z_{j+1}$ for $i \neq j$, and for measuring $\prod_j \widebar Z_j = Z_1 Z_n$ and $\prod_{j \neq i} \widebar Z_j = Z_1 Z_{i+1}$.  To projectively measure other parities, use a two-ancilla circuit as in \figref{f:6qubitflaggedextraction}.  Of course, symmetrical circuits work for measuring logical $X$ operators.

\bibliography{q}

\begin{thebibliography}{19}%
\makeatletter
\providecommand \@ifxundefined [1]{%
 \@ifx{#1\undefined}
}%
\providecommand \@ifnum [1]{%
 \ifnum #1\expandafter \@firstoftwo
 \else \expandafter \@secondoftwo
 \fi
}%
\providecommand \@ifx [1]{%
 \ifx #1\expandafter \@firstoftwo
 \else \expandafter \@secondoftwo
 \fi
}%
\providecommand \natexlab [1]{#1}%
\providecommand \enquote  [1]{``#1''}%
\providecommand \bibnamefont  [1]{#1}%
\providecommand \bibfnamefont [1]{#1}%
\providecommand \citenamefont [1]{#1}%
\providecommand \href@noop [0]{\@secondoftwo}%
\providecommand \href [0]{\begingroup \@sanitize@url \@href}%
\providecommand \@href[1]{\@@startlink{#1}\@@href}%
\providecommand \@@href[1]{\endgroup#1\@@endlink}%
\providecommand \@sanitize@url [0]{\catcode `\\12\catcode `\$12\catcode
  `\&12\catcode `\#12\catcode `\^12\catcode `\_12\catcode `\%12\relax}%
\providecommand \@@startlink[1]{}%
\providecommand \@@endlink[0]{}%
\providecommand \url  [0]{\begingroup\@sanitize@url \@url }%
\providecommand \@url [1]{\endgroup\@href {#1}{\urlprefix }}%
\providecommand \urlprefix  [0]{URL }%
\providecommand \Eprint [0]{\href }%
\providecommand \doibase [0]{http://dx.doi.org/}%
\providecommand \selectlanguage [0]{\@gobble}%
\providecommand \bibinfo  [0]{\@secondoftwo}%
\providecommand \bibfield  [0]{\@secondoftwo}%
\providecommand \translation [1]{[#1]}%
\providecommand \BibitemOpen [0]{}%
\providecommand \bibitemStop [0]{}%
\providecommand \bibitemNoStop [0]{.\EOS\space}%
\providecommand \EOS [0]{\spacefactor3000\relax}%
\providecommand \BibitemShut  [1]{\csname bibitem#1\endcsname}%
\let\auto@bib@innerbib\@empty
\bibitem [{\citenamefont {Linke}\ \emph {et~al.}(2016)\citenamefont {Linke},
  \citenamefont {Gutierrez}, \citenamefont {Landsman}, \citenamefont {Figgatt},
  \citenamefont {Debnath}, \citenamefont {Brown},\ and\ \citenamefont
  {Monroe}}]{LinkeMonroe16errordetectionexperiment}%
  \BibitemOpen
  \bibfield  {author} {\bibinfo {author} {\bibfnamefont {N.~M.}\ \bibnamefont
  {Linke}}, \bibinfo {author} {\bibfnamefont {M.}~\bibnamefont {Gutierrez}},
  \bibinfo {author} {\bibfnamefont {K.~A.}\ \bibnamefont {Landsman}}, \bibinfo
  {author} {\bibfnamefont {C.}~\bibnamefont {Figgatt}}, \bibinfo {author}
  {\bibfnamefont {S.}~\bibnamefont {Debnath}}, \bibinfo {author} {\bibfnamefont
  {K.~R.}\ \bibnamefont {Brown}}, \ and\ \bibinfo {author} {\bibfnamefont
  {C.}~\bibnamefont {Monroe}},\ }\href@noop {} {\enquote {\bibinfo {title}
  {Experimental demonstration of quantum fault tolerance},}\ } (\bibinfo {year}
  {2016}),\ \Eprint {http://arxiv.org/abs/1611.06946}
  {arXiv:1611.06946 [quant-ph]} \BibitemShut {NoStop}%
\bibitem [{\citenamefont {Gottesman}(2016)}]{Gottesman16smallexperiments}%
  \BibitemOpen
  \bibfield  {author} {\bibinfo {author} {\bibfnamefont {Daniel}\ \bibnamefont
  {Gottesman}},\ } {\enquote {\bibinfo
  {title} {Quantum fault tolerance in small experiments},}} (\bibinfo {year}
  {2016}),\ \Eprint {http://arxiv.org/abs/1610.03507}
  {arXiv:1610.03507 [quant-ph]} \BibitemShut {NoStop}%
\bibitem [{\citenamefont {Shor}(1996)}]{Shor96}%
  \BibitemOpen
  \bibfield  {author} {\bibinfo {author} {\bibfnamefont {Peter~W.}\
  \bibnamefont {Shor}},\ }\bibfield  {title} {\enquote {\bibinfo {title}
  {Fault-tolerant quantum computation},}\ }in\ \href {\doibase 10.1109/SFCS.1996.548464} {\emph {\bibinfo {booktitle} {Proc. 37th Symp. on
  Foundations of Computer Science (FOCS)}}}\ (\bibinfo {year} {1996})\
  p.~\bibinfo {pages} {96},\ \Eprint
  {http://arxiv.org/abs/quant-ph/9605011} {arXiv:quant-ph/9605011}
  \BibitemShut {NoStop}%
\bibitem [{\citenamefont {DiVincenzo}\ and\ \citenamefont
  {Aliferis}(2007)}]{DiVincenzoAliferis06slow}%
  \BibitemOpen
  \bibfield  {author} {\bibinfo {author} {\bibfnamefont {David~P.}\
  \bibnamefont {DiVincenzo}}\ and\ \bibinfo {author} {\bibfnamefont {Panos}\
  \bibnamefont {Aliferis}},\ }\bibfield  {title} {\enquote {\bibinfo {title}
  {Effective fault-tolerant quantum computation with slow measurements},}\
  }\href {\doibase 10.1103/PhysRevLett.98.020501} {\bibfield  {journal}
  {\bibinfo  {journal} {Phys. Rev. Lett.}\ }\textbf {\bibinfo {volume} {98}},\
  \bibinfo {pages} {220501} (\bibinfo {year} {2007})},\ \Eprint
  {http://arxiv.org/abs/quant-ph/0607047} {arXiv:quant-ph/0607047}
  \BibitemShut {NoStop}%
\bibitem [{\citenamefont {Steane}(1997)}]{Steane97}%
  \BibitemOpen
  \bibfield  {author} {\bibinfo {author} {\bibfnamefont {Andrew~M.}\
  \bibnamefont {Steane}},\ }\bibfield  {title} {\enquote {\bibinfo {title}
  {Active stabilization, quantum computation, and quantum state synthesis},}\
  }\href {\doibase 10.1103/PhysRevLett.78.2252} {\bibfield  {journal} {\bibinfo
   {journal} {Phys. Rev. Lett.}\ }\textbf {\bibinfo {volume} {78}},\ \bibinfo
  {pages} {2252--2255} (\bibinfo {year} {1997})},\ \Eprint
  {http://arxiv.org/abs/quant-ph/9611027} {arXiv:quant-ph/9611027}
  \BibitemShut {NoStop}%
\bibitem [{\citenamefont {Steane}(2002)}]{Steane02}%
  \BibitemOpen
  \bibfield  {author} {\bibinfo {author} {\bibfnamefont {Andrew~M.}\
  \bibnamefont {Steane}},\ }
  {\enquote {\bibinfo {title} {Fast fault-tolerant filtering of quantum codewords},}} (\bibinfo {year} {2002}),\ \Eprint
  {http://arxiv.org/abs/quant-ph/0202036} {arXiv:quant-ph/0202036}
  \BibitemShut {NoStop}%
\bibitem [{\citenamefont {Knill}(2005{\natexlab{a}})}]{Knill03erasure}%
  \BibitemOpen
  \bibfield  {author} {\bibinfo {author} {\bibfnamefont {Emanuel}\ \bibnamefont
  {Knill}},\ }\bibfield  {title} {\enquote {\bibinfo {title} {Scalable quantum
  computing in the presence of large detected-error rates},}\ }\href {\doibase 10.1103/PhysRevA.71.042322} {\bibfield  {journal} {\bibinfo  {journal} {Phys.
  Rev. A}\ }\textbf {\bibinfo {volume} {71}},\ \bibinfo {pages} {042322}
  (\bibinfo {year} {2005}{\natexlab{a}})},\ \Eprint
  {http://arxiv.org/abs/quant-ph/0312190} {arXiv:quant-ph/0312190}
  \BibitemShut {NoStop}%
\bibitem [{\citenamefont {Yoder}\ and\ \citenamefont
  {Kim}(2017)}]{YoderKim16trianglecodes}%
  \BibitemOpen
  \bibfield  {author} {\bibinfo {author} {\bibfnamefont {Theodore~J.}\
  \bibnamefont {Yoder}}\ and\ \bibinfo {author} {\bibfnamefont {Isaac~H.}\
  \bibnamefont {Kim}},\ }\bibfield  {title} {\enquote {\bibinfo {title} {The
  surface code with a twist},}\ }\href {\doibase 10.22331/q-2017-04-25-2}
  {\bibfield  {journal} {\bibinfo  {journal} {Quantum}\ }\textbf {\bibinfo
  {volume} {1}},\ \bibinfo {pages} {2} (\bibinfo {year} {2017})},\ \Eprint
  {http://arxiv.org/abs/1612.04795} {arXiv:1612.04795
  [quant-ph]} \BibitemShut {NoStop}%
\bibitem [{\citenamefont {Chao}\ and\ \citenamefont
  {Reichardt}(2017)}]{ChaoReichardt17automorphisms}%
  \BibitemOpen
  \bibfield  {author} {\bibinfo {author} {\bibfnamefont {Rui}\ \bibnamefont
  {Chao}}\ and\ \bibinfo {author} {\bibfnamefont {Ben~W.}\ \bibnamefont
  {Reichardt}},\ }\href@noop {} {\enquote {\bibinfo {title} {Fault-tolerant
  computation with few qubits},}\ } (\bibinfo {year} {2017}),\ \bibinfo {note}
  {in preparation}\BibitemShut {NoStop}%
\bibitem [{\citenamefont {Stephens}(2014)}]{Stephens14colorcodeft}%
  \BibitemOpen
  \bibfield  {author} {\bibinfo {author} {\bibfnamefont {Ashley~M.}\
  \bibnamefont {Stephens}},\ }\bibfield  {title} {\enquote {\bibinfo {title}
  {Efficient fault-tolerant decoding of topological color codes},}\ }\href@noop
  {} {\  (\bibinfo {year} {2014})},\ \Eprint
  {http://arxiv.org/abs/1402.3037} {arXiv:1402.3037
  [quant-ph]} \BibitemShut {NoStop}%
\bibitem [{\citenamefont {Laflamme}\ \emph {et~al.}(1996)\citenamefont
  {Laflamme}, \citenamefont {Miquel}, \citenamefont {Paz},\ and\ \citenamefont
  {Zurek}}]{LaflammeMiquelPazZurek96}%
  \BibitemOpen
  \bibfield  {author} {\bibinfo {author} {\bibfnamefont {Raymond}\ \bibnamefont
  {Laflamme}}, \bibinfo {author} {\bibfnamefont {Cesar}\ \bibnamefont
  {Miquel}}, \bibinfo {author} {\bibfnamefont {Juan~Pablo}\ \bibnamefont
  {Paz}}, \ and\ \bibinfo {author} {\bibfnamefont {Wojciech~Hubert}\
  \bibnamefont {Zurek}},\ }\bibfield  {title} {\enquote {\bibinfo {title}
  {Perfect quantum error correcting code},}\ }\href {\doibase 10.1103/PhysRevLett.77.198} {\bibfield  {journal} {\bibinfo  {journal} {Phys.
  Rev. Lett.}\ }\textbf {\bibinfo {volume} {77}},\ \bibinfo {pages} {198--201}
  (\bibinfo {year} {1996})},\ \Eprint
  {http://arxiv.org/abs/quant-ph/9602019} {arXiv:quant-ph/9602019}
  \BibitemShut {NoStop}%
\bibitem [{\citenamefont {Gottesman}(2010)}]{Gottesman09faulttolerance}%
  \BibitemOpen
  \bibfield  {author} {\bibinfo {author} {\bibfnamefont {Daniel}\ \bibnamefont
  {Gottesman}},\ }\bibfield  {title} {\enquote {\bibinfo {title} {An
  introduction to quantum error correction and fault-tolerant quantum
  computation},}\ }in\ \href {\doibase 10.1090/psapm/068} {\emph {\bibinfo
  {booktitle} {Quantum Information Science and Its Contributions to
  Mathematics, Proc. Symp. in Applied Mathematics}}},\ Vol.~\bibinfo {volume}
  {68}\ (\bibinfo  {publisher} {Amer. Math. Soc., Providence, {RI}},\ \bibinfo
  {year} {2010})\ pp.\ \bibinfo {pages} {13--58},\ \Eprint
  {http://arxiv.org/abs/0904.2557} {arXiv:0904.2557
  [quant-ph]} \BibitemShut {NoStop}%
\bibitem [{\citenamefont {Steane}(1996{\natexlab{a}})}]{Steane96css}%
  \BibitemOpen
  \bibfield  {author} {\bibinfo {author} {\bibfnamefont {Andrew~M.}\
  \bibnamefont {Steane}},\ }\bibfield  {title} {\enquote {\bibinfo {title}
  {Error correcting codes in quantum theory},}\ }\href {\doibase 10.1103/PhysRevLett.77.793} {\bibfield  {journal} {\bibinfo  {journal} {Phys.
  Rev. Lett.}\ }\textbf {\bibinfo {volume} {77}},\ \bibinfo {pages} {793--797}
  (\bibinfo {year} {1996}{\natexlab{a}})}\BibitemShut {NoStop}%
\bibitem [{\citenamefont {Calderbank}\ \emph {et~al.}(1997)\citenamefont
  {Calderbank}, \citenamefont {Rains}, \citenamefont {Shor},\ and\
  \citenamefont {Sloane}}]{CalderbankRainsShorSloane96smallcodes}%
  \BibitemOpen
  \bibfield  {author} {\bibinfo {author} {\bibfnamefont {A.~R.}\ \bibnamefont
  {Calderbank}}, \bibinfo {author} {\bibfnamefont {Eric~M.}\ \bibnamefont
  {Rains}}, \bibinfo {author} {\bibfnamefont {Peter~W.}\ \bibnamefont {Shor}},
  \ and\ \bibinfo {author} {\bibfnamefont {N.~J.~A.}\ \bibnamefont {Sloane}},\
  }\bibfield  {title} {\enquote {\bibinfo {title} {Quantum error correction and
  orthogonal geometry},}\ }\href {\doibase 10.1103/PhysRevLett.78.405}
  {\bibfield  {journal} {\bibinfo  {journal} {Phys. Rev. Lett.}\ }\textbf
  {\bibinfo {volume} {78}},\ \bibinfo {pages} {405--408} (\bibinfo {year}
  {1997})},\ \Eprint {http://arxiv.org/abs/quant-ph/9605005}
  {arXiv:quant-ph/9605005} \BibitemShut {NoStop}%
\bibitem [{\citenamefont {Gottesman}(1996)}]{Gottesman96codes}%
  \BibitemOpen
  \bibfield  {author} {\bibinfo {author} {\bibfnamefont {Daniel}\ \bibnamefont
  {Gottesman}},\ }\bibfield  {title} {\enquote {\bibinfo {title} {Class of
  quantum error-correcting codes saturating the quantum {H}amming bound},}\
  }\href {\doibase 10.1103/PhysRevA.54.1862} {\bibfield  {journal} {\bibinfo
  {journal} {Phys. Rev. A}\ }\textbf {\bibinfo {volume} {54}},\ \bibinfo
  {pages} {1862} (\bibinfo {year} {1996})},\ \Eprint
  {http://arxiv.org/abs/quant-ph/9604038} {arXiv:quant-ph/9604038}
  \BibitemShut {NoStop}%
\bibitem [{\citenamefont {Steane}(1996{\natexlab{b}})}]{Steane96codes}%
  \BibitemOpen
  \bibfield  {author} {\bibinfo {author} {\bibfnamefont {Andrew~M.}\
  \bibnamefont {Steane}},\ }\bibfield  {title} {\enquote {\bibinfo {title}
  {Simple quantum error-correcting codes},}\ }\href {\doibase 10.1103/PhysRevA.54.4741} {\bibfield  {journal} {\bibinfo  {journal} {Phys.
  Rev. A}\ }\textbf {\bibinfo {volume} {54}},\ \bibinfo {pages} {4741}
  (\bibinfo {year} {1996}{\natexlab{b}})},\ \Eprint
  {http://arxiv.org/abs/quant-ph/9605021} {arXiv:quant-ph/9605021}
  \BibitemShut {NoStop}%
\bibitem [{\citenamefont {Grassl}(2007)}]{Grassl07codetable}%
  \BibitemOpen
  \bibfield  {author} {\bibinfo {author} {\bibfnamefont {Markus}\ \bibnamefont
  {Grassl}},\ }\href@noop {} {\enquote {\bibinfo {title} {Bounds on the minimum
  distance of linear codes and quantum codes},}\ }\bibinfo {howpublished}
  {Online available at \url{http://www.codetables.de}} (\bibinfo {year}
  {2007})\BibitemShut {NoStop}%
\bibitem [{\citenamefont {Knill}(2005{\natexlab{b}})}]{Knill05}%
  \BibitemOpen
  \bibfield  {author} {\bibinfo {author} {\bibfnamefont {Emanuel}\ \bibnamefont
  {Knill}},\ }\bibfield  {title} {\enquote {\bibinfo {title} {Quantum computing
  with realistically noisy devices},}\ }\href {\doibase 10.1038/nature03350}
  {\bibfield  {journal} {\bibinfo  {journal} {Nature}\ }\textbf {\bibinfo
  {volume} {434}},\ \bibinfo {pages} {39--44} (\bibinfo {year}
  {2005}{\natexlab{b}})}\BibitemShut {NoStop}%
\bibitem [{\citenamefont {Cross}\ \emph {et~al.}(2009)\citenamefont {Cross},
  \citenamefont {DiVincenzo},\ and\ \citenamefont
  {Terhal}}]{CrossDiVincenzoTerhal09codes}%
  \BibitemOpen
  \bibfield  {author} {\bibinfo {author} {\bibfnamefont {Andrew~W.}\
  \bibnamefont {Cross}}, \bibinfo {author} {\bibfnamefont {David~P.}\
  \bibnamefont {DiVincenzo}}, \ and\ \bibinfo {author} {\bibfnamefont
  {Barbara~M.}\ \bibnamefont {Terhal}},\ }\bibfield  {title} {\enquote
  {\bibinfo {title} {A comparative code study for quantum fault-tolerance},}\
  }\href@noop {} {\bibfield  {journal} {\bibinfo  {journal} {Quant. Inf.
  Comput.}\ }\textbf {\bibinfo {volume} {9}},\ \bibinfo {pages} {541--572}
  (\bibinfo {year} {2009})},\ \Eprint {http://arxiv.org/abs/0711.1556} {arXiv:0711.1556 [quant-ph]} \BibitemShut {NoStop}%
\end{thebibliography}%

\end{document}